\documentclass[10pt]{amsart}

\usepackage{cite}

\usepackage{kotex}
\usepackage{lineno,hyperref}
\modulolinenumbers[5]
\usepackage{epstopdf}
\usepackage{pifont}
\usepackage{latexsym}
\usepackage{graphics}
\usepackage{graphicx}
\usepackage{float}
\usepackage[dvips]{xy}
\xyoption{all}
\usepackage{ifpdf}
\usepackage{multirow}
\usepackage{amssymb, amsthm, amsmath}
\usepackage{hhline}
\usepackage{centernot}
\usepackage{verbatim, comment, color}
\usepackage{enumerate}
\usepackage{setspace}
\usepackage{lipsum}
\usepackage{subcaption}

\newtheorem{theorem}{Theorem} 
\newtheorem{proposition}[theorem]{Proposition}
\newtheorem{lemma}[theorem]{Lemma}
\newtheorem{remark}[theorem]{Remark}
\newtheorem{corollary}[theorem]{Corollary}
\newtheorem{definition}[theorem]{Definition}

\newcommand{\ba}{\begin{align}}
\newcommand{\ea}{\end{align}}  

\newcommand{\be}{\begin{equation}}
\newcommand{\ee}{\end{equation}}
\newcommand{\bea}{\begin{eqnarray}}
\newcommand{\eea}{\end{eqnarray}}
\newcommand{\barr}{\begin{array}}
\newcommand{\earr}{\end{array}}
\newcommand{\bn}{\begin{enumerate}}
\newcommand{\en}{\end{enumerate}}
\newcommand{\bi}{\begin{itemize}}
\newcommand{\ei}{\end{itemize}}
\newcommand{\bbbm}{\begin{pmatrix}}
\newcommand{\eeem}{\end{pmatrix}}

\newcommand{\bbN}{{\bf N}}

\newcommand{\cE}{{\cal E}}

\newcommand{\R}{{\mathbf R}}

\newcommand{\al}{\alpha}
\newcommand{\bt}{\beta}

\newcommand{\ep}{\epsilon}

\newcommand{\la}{\lambda}

\newcommand{\ignore}[1]{}{}
\newcommand{\noin}{\noindent}

\newcommand{\nn}{\nonumber}

\newcommand{\p}{{\partial}}

\newcommand{\q}{\quad}

 \newcommand{\Id}{\mathop{\rm Id}}
 \newcommand{\mab}{{\mu_{\alpha,\beta}}}
\newcommand{\mi}{{\hat \mu}}
\newcommand{\Mi}{{\hat M_0}}
\newcommand{\Miw}{{\hat M}}
\newcommand{\Prob}{{\mathcal P}}
\newcommand{{\QED}}{{\hfill QED} \smallskip}
\newcommand{\Rab}{R_{\alpha,\beta}}

\newcommand{\Rn}{{\R^n}}
\newcommand{\spt}{\mathop{\rm spt}}

\newcommand{\Vab}{{V_{\al,\beta}}}
\newcommand{\Vib}{{V_{\infty,\beta}}}
\newcommand{\Wa}{W_\alpha}
\newcommand{\Wb}{{W_\beta}}
\newcommand{\Wab}{W_{\alpha,\beta}}
\newcommand{\wab}{w_{\alpha,\beta}}
\newcommand{\wib}{w_{\infty,\beta}}
\newcommand{\was}{w_{\alpha,\alpha^*}}
\newcommand{\wsb}{w_{\alpha^*,\beta}}

\newcommand{\Wib}{W_{\infty,\beta}}
\newcommand{\Wit}{{W_{\infty,2}}}
\newcommand{\Wt}{{W_{2}}}

\renewcommand{\subset}{\subseteq}

\renewcommand{\phi}{\varphi}
\newcommand{\cal}{\mathcal}

\DeclareMathOperator{\conv}{conv}
  \DeclareMathOperator*{\diam}{diam}

 \DeclareMathOperator*{\argmin}{argmin}
  
\DeclareMathOperator*{\Var}{{Var}}
\DeclareMathOperator*{\bary}{{\bar x}}

\numberwithin{equation}{section}
\numberwithin{theorem}{section}
\begin{document}
\title
[When do particles self-assemble into a regular simplex?] {
Isodiametry, variance, and regular simplices from particle interactions
}
\thanks{\em  TL is grateful for the support of ShanghaiTech University, and in addition, to the University of Toronto and its Fields Institute for the Mathematical
Sciences, where parts of this work were performed.  RM  acknowledges partial support of his research by
the Canada Research Chairs Program and Natural Sciences and Engineering Research Council of Canada Grants 217006-15 and -20.
{The authors are grateful to Andrea Bertozzi, Almut Burchard, {Tomasz Tkocz
 and an anonymous seminar participant at Seoul National University
 for stimulating interactions,} and to Hyejung Choi for drawing the figures.
}
\copyright 2020 by the authors.
}

\date{\today}

\author{Tongseok Lim and Robert J. McCann}
\address{Tongseok Lim: Krannert School of Management \newline  Purdue University, West Lafayette, Indiana 47907}
\email{lim336@purdue.edu}
\address{Robert J. McCann: Department of Mathematics \newline University of Toronto, Toronto ON Canada}
\email{mccann@math.toronto.edu}

\begin{abstract}
Consider a collection of particles 
interacting through an attractive-repulsive potential given as a difference of power laws 
and normalized so that its unique minimum occurs at unit separation.  For a range of exponents corresponding to mild repulsion and strong attraction, we show
that  the minimum energy configuration is uniquely attained 
 --- apart from translations and rotations --- by equidistributing
the particles over the vertices of a regular top-dimensional simplex (i.e. an equilateral triangle in two dimensions and regular tetrahedron in three).
If the attraction is not assumed to be strong, we show these configurations are at least local energy minimizers in the relevant
$d_\infty$ metric from optimal transportation,  as are all of the other 
uncountably many unbalanced configurations with the same support.  
We infer the existence of phase transitions.

The proof is based in part on a simple isodiametric variance bound which
characterizes regular simplices:  it shows that among probability measures on $\R^n$ whose supports
have at most unit diameter,  the variance around the mean is maximized precisely by 
those measures which assign mass $1/(n+1)$ to each vertex of a (unit-diameter) regular simplex.
\end{abstract}

\maketitle
\noindent\emph{Keywords: isodiametric variance bound, interacting particles,
regular simplex, aggregation equation, self-assembly, repulsive-attractive,  power law potentials, 
$L^\infty$-Kantorovich-Rubinstein-Wasserstein metric, mildly repulsive, 
self-organizing, pattern formation, calculus of variations,  symmetry breaking,  
interaction energy, Jung's inequality}

\noindent\emph{MSC2010 Classification 70F45, 35Q92, 49S05, 52C17, 91D25, 92D50 }
\section{Introduction}

The energy of {a collection of interacting} particles with mass distribution $d\mu(x) \ge 0$ on $\R^n$ is given by 
\ba \label{energy}
 \cE_W(\mu) = \iint_{\R^n \times \R^n} W(x-y) d\mu(x)d\mu(y),
\end{align}
assuming the particles interact with each other through a pair potential $W(x)$.   Normalizing
the  collection of particles to have unit mass ensures that $\mu$ belongs to the space $\Prob(\R^n)$ of Borel
probability measures on $\R^n$. 

Our goal is to identify local and global energy minimizers
of $\cE_W(\mu)$ on $\Prob(\Rn)$,  for {\em power-law} potentials $W=\Wab$ where
\begin{align} \label{potential}
\Wa &:= |x|^\al/\al  \q {\rm and}
\\ \Wab(x) &:= \Wa(x) - \Wb(x) \q  \alpha \ge \beta>-n
\label{potential2}
\end{align}
is of {\em attractive-repulsive} type $\al>\beta$; here $\al$ is the exponent of attraction,
$\beta$ is the exponent of repulsion, 
and we have chosen units of length 
so that $\Wab$ is minimized precisely on the unit sphere $|x|=1$.  
The Lennard-Jones potentials \cite{Lennard-Jones24} 
 fall into this class, 
including  $(\al,\beta)=(-6,-12)$, except that we will be concerned almost exclusively
with power laws having positive rather than negative exponents, particularly 
those in the {\em mildly repulsive} triangle $\al > \beta \ge 2$ 
investigated by the quartet and trio composed of Balagu\'e, Carrillo, Laurent and Raoul~\cite{BalagueCarrilloLaurentRaoul13} and Carrillo, Figalli and Patacchini~\cite{CarrilloFigalliPatacchini17} respectively.
The term {\em mildly repulsive} reflects the fact that $W$ flattens out around the origin (and the Hausdorff dimension of the support of the minimizer 
decreases~\cite{BalagueCarrilloLaurentRaoul13}) as $\beta$ increases.
We shall be particularly interested in the behaviour of the problem on the boundary of the mildly repulsive triangle:
this consists of three lines which we call the {\em hard confinement} limit $\alpha = +\infty$,
the {\em centrifugal line} $\beta =2$ and the {\em null line} $\alpha=\beta$, on which the energy is identically zero.  (The line $\alpha=2$ is also distinguished;
for reasons explained below we call it the {\em centripetal line} even though it lies outside our triangle of interest.)

Our first result concerns behaviour near the hard confinement limit.
For each $\beta \ge 2$,
if $\al$ is sufficiently large it asserts the energy \eqref{energy} is uniquely minimized on $\Prob(\R^n)$ by measures $\mu$ which
equidistribute their mass over the vertices of a unit-diameter regular simplex. 
This confirms a phenomenon which has often been observed in  
dynamical simulations
\cite{AlbiBalagueCarrilloVonBrecht14} 
\cite{BalagueCarrilloLaurentRaoul13} \cite{BertozziKolokolnikovSunUminskyVonBrecht15} 
\cite{CraigBertozzi16}
yet has largely defied explanation. 
Apart from results in one-dimension due to Kang, Kim, Lim and Seo~\cite{KangKimLimSeo19p} and  their references,  
the best understanding to date of this mildly repulsive phenomenology
comes from work of the 
quartet~\cite{BalagueCarrilloLaurentRaoul13}, who established that local minimizers
vanish outside a countable set,  and the 
trio~\cite{CarrilloFigalliPatacchini17}, who gave a geometric restriction
on the shape of this support which translated into a bound on the number of points it contains in the case of global minimizers,
and which we can now replace with its sharp value $n+1$ at least in the range of validity of our results.

The behaviour we describe is very different from what happens when the repulsion is stronger \cite{FetecauHuangKolokolnikov11} 
\cite{FetecauHuang13}: when $\beta \in (-n,2]$, 
the functional \eqref{energy} admits spherically symmetric critical points given by densities if either $\al$ or $\beta$ is even \cite{CarrilloHuang17}
or if $\al<0$ \cite{ChoksiFetecauTopaloglu15};
some of these are conjectured 
to be global energy minimizers --- a conjecture which has been proven at the point $(\al,\beta) = (2,2-n)$ where Newtonian repulsion competes with centripetal attraction 
by Choksi, Fetecau and Topaloglu \cite{ChoksiFetecauTopaloglu15},
and which 
follows from the convexity established by Lopes \cite{Lopes19} in the larger rectangle $(\al,\beta) \in [2,4] \times (-n,0)$ whose left boundary is the centripetal line. 
Even in two dimensions a 
wide variety of behaviours interpolating between this regime and ours 
has been reported by, e.g.,~Kolokolnikov, Uminsky and Bertozzi with Sun \cite{KolokolnikovSunUminskyBertozzi11} and with von Brecht 
\cite{vonBrechtUminskyKolokolnikovBertozzi12}.
Very recently,  the analogous problem has been studied under an incompressiblity constraint imposed by a uniform bound on the density of $\mu$  \cite{BurchardChoksiTopaloglu18}: 
Frank and Lieb \cite{FrankLieb18} established the presence of a phase transition as the bound is varied;
 it is in this context that the work of Lopes is set.  A few subsequent developments concerning nonlocal interaction energies can be found in Frank and Lieb \cite{FrankLieb19}
 and Delgadino, Yan and Yao \cite{DelgadinoYanYao20}.

Much of the interest in minimizers of the functional \eqref{energy} stems
from the fact that it is a Lyapunov functional \cite{CarrilloMcCannVillani03} \cite{CarrilloFigalliPatacchini17}
for the {\em self-assembly} or {\em aggregation} equation \cite{MogilnerEdelstein-Keshet99}
\ba \label{dynamics}
\frac{\p \mu}{\p t} = \nabla \cdot (\mu \nabla W* \mu),
\end{align}
modeling dissipation-dominated 
dynamics for a large number of particles interacting
through the pair potential $W$; see e.g.~\cite{CarrilloChoiHauray14} and the references there. Families of local
energy minimizers of \eqref{energy} therefore form stable manifolds
for the dynamics \eqref{dynamics}.  The shape of $W$ has been chosen so that it is energetically favorable 
for particles to try to position themselves at unit distance apart,  to the extent this is feasible given the large number of 
particles. Dynamics analogous to \eqref{dynamics} have been proposed as models for the kinetic flocking and swarming behaviour of biological organisms
\cite{MogilnerEdelstein-Keshet99} \cite{TopazBertozziLewis06}, self-assembly and condensation of granular media \cite{Toscani00} and nanomaterials \cite{HolmPutkaradze06},
and even strategies in game theory \cite{BlanchetCarlier14}.

The fact that the minimizers we describe break the rotational symmetry of the functional \eqref{energy} already 
suggests that the problem is unlikely to yield to the usual convexity or symmetrization techniques from the calculus of variations 
\cite{Kawohl85} \cite{McCann97}  \cite{BorweinZhu13} \cite{CarrilloHittmeirVolzoneYao16p}.  Instead we extend the definition \eqref{potential} to $\alpha =+\infty$ by setting
\[
W_\infty (x):= \lim_{\alpha \to \infty} \Wa(x) 
\]
so that 
\begin{align} \label{i-potential}
\Wib(x):=
\left\{
\begin{array}{cc}
-\Wb(x) 
\text{ if } \  |x| \le 1, \\
+\infty \  \text{ if } \  |x| > 1,
\end{array}
\right.
\end{align}
and work perturbatively around this 
{hard confinement} limit, for which we analyze the minimization problem
\begin{align} \label{problem}
\min_{\mu \in {\cal P}(\R^n)} \cE_{\Wab} (\mu), \q \infty \ge \al \ge \beta \ge 2
\end{align}
by comparing it to the corner case $(\al,\beta)=(\infty,2)$ where hard confinement meets the centrifugal line.
Such an approach to the more repulsive regime $\beta <0$ with an incompressibility constraint 
was also suggested by Burchard, Choksi and Topaloglu \cite{BurchardChoksiTopaloglu18}, 
 and subsequently pursued by Burchard, Choksi and Hess-Childs in parallel with the present work \cite{BurchardChoksiHess-Childs20}.
What distinguishes the centrifugal (respectively centripetal) line is that,
for probability measures $\mu \in \Prob(\R^n)$ with second moments,
the elementary calculation \ba\label{variance}
\cE_{\Wt}(\mu) 
&= \int_{\R^n} |x|^2 d\mu(x) - |\bary(\mu)|^2 =: \Var(\mu),
\\  \text{where}\ \bary(\mu) &:=  \int_{\R^n} x d\mu(x) \q \text{is the barycenter of}\ \mu,
\label{barycenter}
\end{align}
shows that the repulsive (respectively attractive) term in the energy reduces to the variance of $\mu$ around its mean, 
 as in e.g.~\cite{ChoksiFetecauTopaloglu15}.   
Moreover, the variance \eqref{variance} becomes a {\em linear} (as opposed to quadratic) function of $\mu$ when restricted to measures 
\be\label{zero mean}
\Prob_0(\R^n) := \{ \mu \in \Prob(\R^n) \mid \int_\Rn |x|^2 d\mu(x) <+\infty \ {\rm and}\ \bary(\mu) = 0\}
\end{equation}
with center of mass at the origin; this restriction costs no generality since the 
energies \eqref{energy} are  invariant under rigid motions of $\mu$.   The contribution of the variance to the total energy
leads to a term in the Euler-Lagrange equation \eqref{Euler-Lagrange} (found e.g.~in \cite{McCann06} \cite{BalagueCarrilloLaurentRaoul13} 
for our problem) representing a force 
either towards or away from the center of mass ---  depending on whether we are on the centripetal or centrifugal line --- 
and growing linearly with the distance.  This is precisely analogous
to the force which appears in a pressureless model of rotating stars (or in a centrifuge) in $n \le 2$ dimensions; see \cite{McCann06} 
and its references.  The analogy breaks down if $n \ge 3$, 
since our force pulls towards a point rather than an axis of rotation,  but the use of the terms centrifugal and centripetal continues to be justified
by their Latin roots.

The corner case $(\al,\beta)=(\infty,2)$ corresponds to maximizing the 
variance of $\mu$ around its center of mass subject to a constraint on the diameter of the {\em support} 
$\spt \mu \subset \R^n$,  meaning the smallest closed set containing the full mass of $\mu$.
The maximum is attained {\em if and only if $\mu$ equidistributes its mass over the vertices of a regular,
unit diameter $n$-simplex.}  As an appendix explains,  this characterization of the simplex follows from 
(and is equivalent to) an old theorem of Jung \cite{Jung01}. 
Unaware of Jung's theorem, we developed a linear programming and convex-duality based proof 
of this characterization in a companion work \cite{LimMcCann19p}, originally circulated together with
the present results in a single manuscript.

Notice this variational characterization of the simplex
already exhibits symmetry-breaking: although the objective functional \eqref{variance} and its
domain are invariant under rigid motions of $\mu$,  
its extremizers fail to be invariant under either translations or rotations {(see figure \ref{fig:simplex})}.  Nevertheless,  the extremizers are unique apart from such rigid motions.
The usual convexity and symmetrization techniques from the calculus of variations do not easily accommodate optimizations which break symmetries
\cite{Kawohl85} \cite{McCann97} \cite{BorweinZhu13}.  
This characterization plays a key role in the proof of our first main result, 
 whose formulation relies on the following definitions:

\begin{definition}[Simplices]\label{D:simplices} 
(a) A set $K \subset \R^n$ is called a {\em top-dimensional simplex} if $K$ has non-empty interior and is the convex hull of $n+1$ points $\{x_0, x_1,...,x_n\}$ in $\R^n$.
\smallskip

 \noin(b) A set $K \subset \R^n$ is called a {\rm regular $ k$-simplex} if it is the convex hull of $ k+1$ points $\{x_0, x_1,...,x_{k}\}$ in $\R^n$ satisfying $|x_i-x_j|=d$ for some $d>0$ and all $0 \le i < j \le {k}$. 
 The points $\{x_0, x_1,...,x_{k}\}$ are called {\em vertices} of the simplex. 
 \smallskip
  
\noin  (c) In particular, it is called a {\em unit $k$-simplex} if $d=1$.
\end{definition}

\begin{remark}[Regular $n$-simplices $K \subset \Rn$ are top-dimensional] 
\label{R:standard simplex}
A regular $n$-simplex with sidelength $d=\sqrt{2}$ is {linearly} isometric to the following {\em standard simplex} in $\R^{n+1}$
\begin{equation}\label{standard simplex}
\Delta^n := \{ a=\{a_1,...,a_{n+1}\} \in [0,1]^{n+1} \mid \sum_{i=1}^{n+1} a_i =1 \}, 
\end{equation}
which can be verified by simple induction on dimension. We shall use this fact tacitly throughout.
\end{remark}

\begin{figure}[h!]
  \centering
  \begin{subfigure}[b]{0.49\linewidth}
    \includegraphics[width=\linewidth]{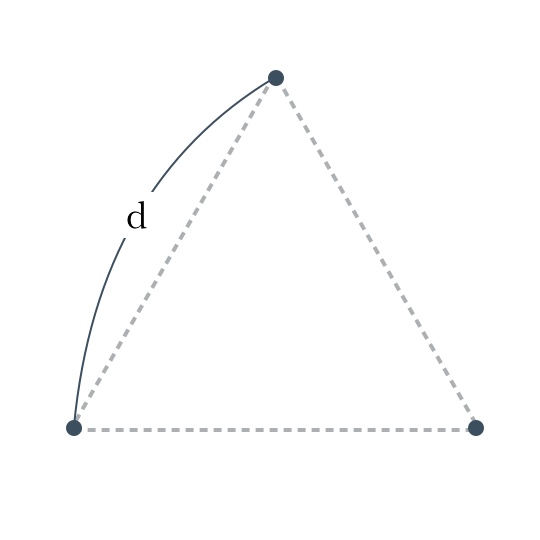}
    \caption{$\spt(\mu)$ in $\R^2$.}
  \end{subfigure}
  \begin{subfigure}[b]{0.49\linewidth}
    \includegraphics[width=\linewidth]{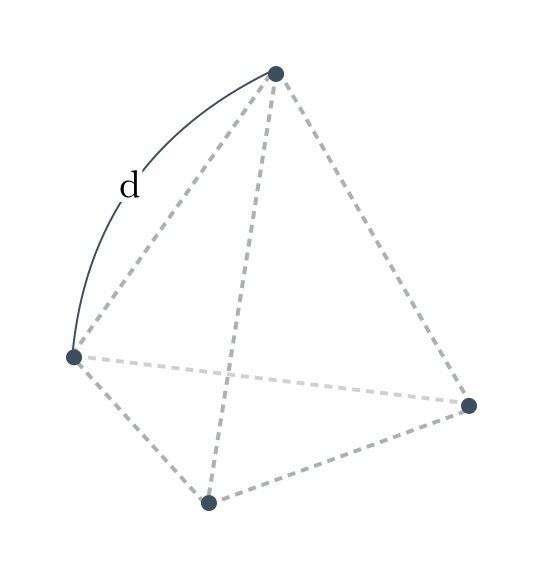}
    \caption{$\spt(\mu)$ in $\R^3$.}
  \end{subfigure}
  \caption{Support of the optimizer $\mu$ in Theorem \ref{T:global}.}
  \label{fig:simplex}
\end{figure}

We can now state our main results.

\begin{theorem}[Mild repulsion with strong attraction is minimized uniquely by the unit $n$-simplex]
\label{T:global} 
Fix $\beta \ge 2$. For all $\al \in [\beta,\infty)$ sufficiently large,
a probability measure $\mu$ minimizes \eqref{problem}  if and only if it is uniformly distributed over the vertices of a unit $n$-simplex. 
\end{theorem}

The following corollary reframes this theorem: 
\begin{corollary}[Phase transition threshold]
\label{C:main} 
For each $\beta \ge 2$ {and  $n \in \bbN$}, there is a minimal value {$\al_{\Delta^n}= \al_{\Delta^n}(\bt) \in [\bt,\infty)$}
such that: for each $\al > \al_{\Delta^n}$,
a probability measure $\mu$ minimizes \eqref{problem}  if and only if $\mu$ assigns mass $1/(n+1)$ to each vertex of a unit $n$-simplex.
\end{corollary}

\noin Proof: For each $\beta \ge 2$ {and $n \in \bbN$}, a minimal $\al_{\Delta^n} \in [\bt,\infty]$ having the stated property obviously exists.
Theorem \ref{T:global} asserts it is finite: $\al_{\Delta^n}<\infty$.
\QED

\begin{remark}[Existence of phase transitions and future directions]\label{R:Lopes}
When $(\al,\beta)=(4,2)$, a result of Lopes \cite{Lopes19} implies that $\cE_{W_{4,2}}(\mu)$ is a convex function of $\mu$.
As a consequence, it must possess at least one spherically symmetric {minimizer},  hence {$\al_{\Delta^n}(2) \ge 4$.}
This establishes a phase transition by showing that the intervals $[2,\alpha_{\Delta^n}(2)]$ and $[\al_{\Delta^n}(2),\infty]$ both have non-empty interiors.
It would be interesting  to understand more about the properties of the threshold function {$\al_{\Delta^n}: [2,\infty)  \to [2,\infty)$,} 
and the behaviour of solutions when $\al$ is at or below the threshold, 
{and similarly of the threshold {$\alpha_{\Delta^n}^{loc}$} for strict local energy minimization when $\beta=2$ introduced at Corollary \ref{C:locastar}.}
  We leave such questions to future research.
\end{remark}

Our second main result concerns local energy minimizers in the Kantorovich-Rubinstein-Wasserstein $d_\infty$ metric
from optimal transportation, whose definition is recalled at \eqref{KRW metric} below.  This is the relevant metric on $\Prob(\Rn)$ for particles
moving at
bounded speeds,  as noted by one of us in \cite{McCann06}, and for the present problem by the quartet \cite{BalagueCarrilloLaurentRaoul13}.

\begin{theorem}[All distributions over unit simplex vertices are $d_\infty$-local energy minimizers] \label{T:local}
Fix $\al>\bt \ge 2$ and any measure $\mi \in \Prob(\Rn)$ whose support $\spt \mi$ coincides with the vertices 
$X =  \{x_0,\ldots,x_n\}$ of a unit $n$-simplex, ordered so that the $m_i := \mi[\{x_i\}]$ are non-decreasing.

If $\bt>2$ or if $\al>  2+  \frac{m_n^2 \min\{n,2\}}{m_0m_1}$,
then there exists $r>0$ such that each $\mu \in \Prob(\Rn)$ with $d_\infty(\mu,\mi)<r$ satisfies
$\cE_{\Wab}(\mu) \ge \cE_{\Wab}(\mi)$, and the inequality is strict unless $\mu$ is a rotated translate of $\mi$. 
\end{theorem}

Since the group of rigid motions has dimension $\frac{n(n+1)}{2}$,  
this theorem provides an uncountable number of $\frac {n(n+1)}2$
\,-\,dimensional manifolds (parameterized by the positive masses $m_0 \le \ldots \le m_n$ assigned 
to each vertex of the simplex) which must be stable under the dynamics
\eqref{dynamics}.  
This both predicts and explains the dynamic formation of unit simplex configurations
observed in simulations throughout the mildly repulsive regime $\al>\bt>2$.
As in one-dimension \cite{KangKimLimSeo19p},
the intuition behind this result is that the configurations described by the theorem are critical points
due to the flatness of the interaction potential $\Wab(x)$ at the origin and at unit distance from it; 
they are stabilized by $\Wab$'s lack of uniform concavity at $x=0$  
in combination with its radially uniform convexity at $|x|=1$ and the geometry of the unit simplex. 

\begin{remark}[Limiting cases {and self-similar aggregation}]
Theorem~\ref{T:local} with $(\bt,m_0)=(2,\frac1{n+1})$ shows the configurations of Theorem \ref{T:global} 
remain $d_\infty$-local energy minimizers for all $\al > 4$ if $n \ge 2$,  and for all $\al>3$ if $n=1$.
For $n=1$,   
versions of both theorems were proved in Kang, Kim, Lim and Seo \cite{KangKimLimSeo19p}
(see also Fellner and Raoul~
\cite{FellnerRaoul10})
along with examples showing in what sense 
the bound on $\al$ required {by Theorem \ref{T:local}} 
is sharp; c.f. Remark \ref{R:1d}.  Studying aggregation with purely attractive power-law potentials,  
Sun, Uminsky and Bertozzi showed that blowing-up solutions can be transformed using
similarity variables into solutions which now appear to interact through an attractive-repulsive potential 
on the centripetal line $\bt=2$. For $n\ge 2$ they then analyze linear stability of two stationary states
for the rescaled dynamics --- (a) the uniform spherically symmetric shell, and (b) the uniform distribution over the vertices of a unit simplex --- to obtain that (a) is linearly stable precisely in the range $\al \in (2,4)$ 
and (b) in the range $\al>4$ \cite{SunUminskyBertozziJMP12}.  
At its end,  their paper raises the questions
of whether these solutions are nonlinearly stable, and whether they are global attractors.  
Theorem \ref{T:local} sheds considerable light on
both questions:  it asserts that (b) is indeed nonlinearly stable when $\al>4$,  but that it cannot be a global attractor since there are also other nonlinearly stable solutions 
(corresponding to $m_0 \ne m_n$).
\end{remark}

\begin{remark}[More general potentials]
 It is natural to expect that the techniques and results of this paper can also be extended to certain more general families of 
potentials which need neither be power-law, spherically symmetric, nor even have attractive-repulsive form globally.   In particular,
the statement of Theorem \ref{T:local} ensures that the same configurations remain $d_\infty$-local minimizers for all
potentials $W$ which agree with $W_{\al,\bt}$ in a neighbourhood in $\R^n$ of radius $2r$ around the $n^2+n+1$ displacements $\{x_i-x_j\}_{0 \le i,j \le n}$
relating vertices of the simplex;
the method of proof also shows a similar result should hold for any $C^2$-smooth radially symmetry potential $W(x)=w(|x|)$
with $w'(0)=w'(1)=0< w''(1)$ and $w''(0)$ not too negative.
Similarly,  for any $C^2$-smooth 
family of radially symmetric potentials $W^{\lambda}(x)=w^{\lambda}(|x|)$, we expect an analog of Theorem \ref{T:global} to hold for $\lambda$ sufficiently large,  if the limit $w^{\infty}$ is attained in a suitable sense
and satisfies
\begin{equation}
w^{\infty}(r) \ge 
\begin{cases}
- r^2 & {\rm if}\ r \le 1
\\ +\infty & {\rm if}\ r >1
\end{cases}
\end{equation}
with equality holding at $r=0$ and $r=1$.
\end{remark}

The paper is organized as follows. Section \ref{S:hard confinement}  recalls 
  our variational characterization of the unit simplex from \cite{LimMcCann19p} 
and  shows the same configurations uniquely
minimize the hard confinement limit $\alpha = +\infty$ of the mildly repulsive energy \eqref{problem}.
Section \ref{S:strong repulsion} introduces the notion of $\Gamma$-convergence with respect to the metrics $d_p$ on probability measures,
and contains a series of preparatory estimates for Section \ref{S:local},
which establishes the presence of $d_\infty$-local minimizers throughout the mildly repulsive triangle $\al>\bt \ge 2$
and extends the characterization of global minimizers from the hard confinement limit
to all sufficiently large values of the attraction exponent~$\alpha$.
{Key estimates of Section \ref{S:strong repulsion} are based on first variation,
whereas those of Section  \ref{S:local} are based on second variation.}
An appendix demonstrates the equivalence of the variational characterization of the unit simplex 
found in our earlier work \cite{LimMcCann19p}
to a classical result of Jung \cite{Jung01}.

\section{Minimizing mild repulsion with hard confinement}
\label{S:hard confinement}
In this section we show that on the entire halfline $\beta \ge 2$ with $\alpha=+\infty$ ---
corresponding to mild repulsion with hard confinement --- the measures which minimize the energy \eqref{problem}
are precisely those which achieve the minimum at its endpoint $(\alpha,\beta) = (\infty,2)$.
At this endpoint, minimizers are given by the 
variational characterization of the unit simplex
proved in our earlier work,
which generalizes to higher dimensions
$n>1$ of a result proven for $n=1$ by Popoviciu \cite{Popoviciu35}:

\begin{theorem}[Isodiametric variance bound and cases of equality \cite{LimMcCann19p}]\label{T:isodiametric variance bound}
If the support of a Borel probability measure $\mu$ on $\R^n$ has diameter no greater than $d$, then $\Var(\mu) \le \frac{n}{2n+2} d^2$.
Equality holds if and only if $\mu$ 
assigns mass $1/(n+1)$ to each vertex of a regular $n$-simplex having diameter $d$.
\end{theorem}

Another proof of this theorem and its 
relation 
to Jung's work \cite{Jung01} are discussed in
Appendix \ref{S:Jung}.
Recall that $\Prob_0(\R^n)$ denotes the set  \eqref{zero mean} of probability measures with second moments and vanishing mean.

\begin{corollary}[Mild repulsion with hard confinement is minimized only by 
unit simplices]\label{C:hard confinement}
Fix $\alpha=+\infty$ and $\beta \ge 2$.
Let $\hat \mu \in \Prob_0(\R^n)$ be a measure which equidistributes its mass over the vertices of
a unit $n$-simplex,  and fix any measure $\mu \in \Prob_0(\R^n)$ which is not a rotation of $\hat \mu$. 
 Then $\cE_{\Wib}(\mu) > \cE_{\Wib}(\hat \mu)$. Thus the minimum \eqref{problem} is uniquely achieved by
 translations and rotations of $\hat \mu$.
\end{corollary}

\noin{\bf Proof:}
Fix any measure $\mu \in \Prob_0(\R^n)$ 
which is not a rotation of $\hat \mu$, and assume 
 $\diam[\spt \mu] \le 1$, since otherwise $\cE_{\Wib}(\mu) = +\infty$ and the inequality holds trivially.
Since  $\beta \ge 2$ and $|x| \le 1$ imply $\beta\Wib(x) \ge 2\Wit (x)$ and equality holds when $|x|=1$,
the uniqueness claim of Theorem \ref{T:isodiametric variance bound}
 asserts 
 \ba
\beta \cE_{\Wib}(\mu)
& \ge  2 \cE_{\Wit}(\mu)  \nn
\\ &> 2 \cE_{\Wit}(\hat \mu)  \nn
\\ &= \beta \cE_{\Wib}(\hat \mu).   \nn
\end{align}
Since $\cE_{\Wab}$ is invariant under rigid motions and its minimizers have bounded diameter \cite{CarrilloFigalliPatacchini17}
(or see \eqref{diameter bound} below),
this shows only $\hat \mu$ and its translations and rotations attain the infimum
 \eqref{problem}.
  \QED

\section{Minimizing mild repulsion with strong attraction}
\label{S:strong repulsion}

We now turn to the question of extending this characterization of energy minimizers to the large finite
values of the attraction exponent $\alpha$ in the mildly repulsive triangle $\alpha \ge \beta \ge 2$.
Recall minimizers $\mu_{\al,\beta}$ of \eqref{problem} are known to exist \cite{ChoksiFetecauTopaloglu15} 
and to satisfy the Euler-Lagrange equation 
\begin{equation}\label{Euler-Lagrange}
\mu * \Wab(x) \ge  \cE_{\Wab}(\mu), \q {\rm with\ equality\ holding}\ \mu-a.e.
\end{equation}
where 
\begin{equation}\label{convolve}
(\mu * W)(x) := \int_\Rn  W(x-y) d\mu(y),
\end{equation}
see e.g. \cite{McCann06}\cite{BalagueCarrilloLaurentRaoul13} or Lemma 2.3 of \cite{CarrilloFigalliPatacchini17};
our normalization $\frac{\delta \cE_{\Wab}}{\delta \mu} = 2 \mu*W$ differs from theirs by a factor of two.
It is not hard to extend this to the hard confinement case $\alpha=+\infty$.
Setting
\ba\label{Mi}
\Miw := \argmin_{\Prob(\Rn)} \cE_{\Wib} \q \text{and} \q
\Mi := \Miw \cap \Prob_0(\Rn),
\end{align}
our strategy is to show a $\Gamma$-convergence result for the $\al \to +\infty$ limit, which implies as in \cite{Braides02} 
that  any sequence of centered minimizers $\mab \in \Prob_0(\Rn)$ 
must approach the $n(n-1)/2$ dimensional manifold $\Mi$
of minimizers for the limiting problem identified in Corollary \ref{C:hard confinement}.
Proposition~\ref{P:convergence of potentials} shows in what sense 
the associated potentials $V_{\al,\beta}:= \mab* \Wab$ converge subsequentially to some $\Vib$. This combines with
the Euler-Lagrange equation \eqref{Euler-Lagrange} to imply all of the mass of $\mab$ must
eventually lie in a small neighbourhood of $\spt \mi_\alpha$ for some $\mi_\alpha \in \Mi$ as a corollary. 

To verify convergence of minimizers to minimizers,  we show the strong attraction problems 
$\Gamma$-converge to the hard confinement problem as $\alpha \to +\infty$ in the 
Kantorovich-Rubinstein-Wassestein metric $d_2$ from optimal transportation \cite{Villani03}.
Recall:

\begin{definition}[$\Gamma$-convergence] 
A sequence $F_i:M \longrightarrow \overline\R$ on a metric space $(M,d)$
 is said to {\em $\Gamma$-converge} to $F_\infty:M \longrightarrow \overline\R$ if (a)
\be\label{Gamma-lsc}
F_\infty(\mu) \le \mathop{\lim\inf}\limits_{i\to \infty} F_i(\mu_i) 
\q {\rm whenever} \q d(\mu_i,\mu) \to 0,
\end{equation}
and (b) each $\mu \in M$ is the limit of a sequence $(\mu_i)_{i} \subset M$ along which
\be\label{Gamma-construction}
F_\infty(\mu) \ge \mathop{\lim\sup}\limits_{i\to \infty} {F_i(\mu_i).  }
\end{equation}

\end{definition}

The main virtue for us of this concept is that it implies $\argmin_M F_i$ cannot have accumulation points as $i\to \infty$ outside of $\argmin_M F_\infty$
\cite{Braides02}.   

For $1 \le p < +\infty$ let
 \[
\Prob_p(\Rn) := \{ \mu \in \Prob(\Rn) \mid \int_{\Rn} |x|^p d\mu(x) <\infty \}
\]
denote the {probability} measures with finite $p$-th moments; let 
$\Prob_\infty(\Rn)$ denote the {probability measures} with bounded support. 
For $\mu,\nu \in \Prob_p(\Rn)$
 define the Kantorovich-Rubinstein-Wasserstein metric 
\be\label{KRW metric}
d_p(\mu,\nu) 
:=\inf_{X \sim \mu, Y \sim \nu} \| X - Y\|_{L^p},
\end{equation}
where the infimum is taken over arbitrary couplings of random variables $X$ and $Y$
whose laws are given by $\mu$ and $\nu$ respectively.  For $p\ne \infty$ the distance $d_p$ is well-known to metrize
narrow convergence (against continuous bounded test functions)
together with convergence of $p$-th moments on $\Prob_p(\Rn)$, e.g.~Theorem 7.12 of \cite{Villani03}.
Fixing $p=2$ hereafter, we endow $\Prob_0(\R^n) \subset \Prob_2(\Rn)$ with the metric $d_2$.

\begin{lemma}[$\Gamma$-convergence to hard confinement]\label{L:Gamma convergence}
{Let} $\al > \beta \ge 2$. The functionals $\cE_{\Wab}$ {$\Gamma$-converge} to $\cE_{\Wib}$ on $(\Prob_2(\Rn),d_2)$ as $\al \to \infty$.
\end{lemma}

\noin
{\bf Proof:}
The construction step \eqref{Gamma-construction} is straightforward: assume $\mu \in \Prob_2(\Rn)$ has $\diam[\spt\mu] \le 1$
since otherwise there is nothing to prove,  and set $\mu_\al := \mu$ for all $\al$.  Since $\Wab$ converges uniformly to $\Wib$ on $|x| \le 1$,
it follows that $\cE_{\Wab}(\mu) \to \cE_{\Wib}(\mu)$ as desired.

To show the `lower semicontinuity' part \eqref{Gamma-lsc} of $\Gamma$-convergence,  suppose 
$$0 =\lim_{\alpha \to \infty }d_2(\mu_\alpha,\mu_\infty) \q{\rm and}\q  
L :=\mathop{\lim\inf}\limits_{\al \to \infty} \cE_{\Wab}(\mu_\al)<+\infty
$$
since otherwise there is nothing to prove.  Choosing a subsequence $\al_i$ along which $\cE_{W_{\al_i,\bt}}(\mu_{\al_i}) \to L$,
we claim
\begin{equation}\label{moment bound}
C :=\mathop{\lim\sup}\limits_{i \to \infty}  \cE_{W_{\al_i}}(\mu_{\al_i})<+\infty.
\end{equation}
We assume the $\mu_{\al_i}$ have compact support  (uniformly in $i$) without loss of generality,  since the general case follows by approximation; i.e.~applying 
the estimate from the remainder of the present paragraph to the normalized restrictions of $\{\mu_{\al_i}\}_i$ to a large ball $B_R(0)$,  and then passing
to the limit $R \to \infty$.  Now for $\alpha > \beta \ge 2$ Jensen's inequality yields
$$
(\beta \cE_{\Wb}(\mu))^{1/\beta} \le (\alpha \cE_{\Wa}(\mu))^{1/\alpha},
$$
whence
$$
\cE_{\Wab} \ge \cE_{\Wa} - \frac1\beta(\al \cE_{\Wa})^{\beta/\al}.
$$
Since $\beta/\al<1$ this implies the desired bound \eqref{moment bound} follows from our hypothesis $\cE_{W_{\al_i,\bt}}(\mu_{\al_i}) \to L$;
in fact $C \le \tilde C$, where $\tilde C=\tilde C(\al,\bt,L)$ is the unique positive number satisfying $L= \tilde C  - (\al  \tilde C)^{\bt/\al}/\bt$.

Having established \eqref{moment bound} (even for sequences of measures with noncompact support),
split $W= W^\le + W^>$ into a short-range and long-range part using
\be\label{srlr}
W^\le(x) := \left \{ 
\begin{array}{cc} 
W(x) &{\rm if} |x| \le 1,
\\  W(\hat e_1) & {\rm else,}
\end{array}\right.
\end{equation}
so that both parts are continuous and $W^\le$ is bounded.
Since $|\Wa^\le| \le 1/\al \to 0$ as $\al \to \infty$, 
we obtain 
\be \mathop{\lim\sup}\limits_{i \to \infty} \cE_{W^>_{\al_i}}(\mu_{\al_i}) = C <\infty
\end{equation}
from  \eqref{moment bound}.
Since $W^>_\beta(x)/W^>_{\al}(x) \to 0$ on $|x|>1$  as $\alpha \to \infty$, 
\be\label{vanishing moments}
\mathop{\lim\sup}\limits_{i \to \infty} \cE_{W^>_{\beta}}(\mu_{\al_i}) = 0 
\end{equation}
follows.
Thus $\diam[\spt \mu_\infty] \le 1$ and \eqref{vanishing moments} also implies
\begin{align*}
\cE_{\Wib}(\mu_\infty) 
&= \cE_{-\Wb}(\mu_\infty)
\\ &= \lim_{i \to \infty} \cE_{-\Wb}(\mu_{\al_i})
\\ &\le L
\end{align*}
as desired.
\QED
\\

Let $\wab(r) := r^\alpha/\alpha - r^\beta/\beta$ be the potential on $\R_+$ for which $\Wab(x) = \wab(|x|)$.
Let $\Rab = (\frac\alpha\beta)^{\frac1{\alpha-\beta}}$ be the unique $R >0$ for which $\wab(R)=0$, and note 
$\Rab \searrow 1$ as $\al \to \infty$.  A second variation calculation by the trio yields 
the following diameter bound, Lemma 2.6 of \cite{CarrilloFigalliPatacchini17}:
\begin{equation}\label{diameter bound}
\diam[\spt \mu] \le \Rab \q {\rm if} \q \mu \in \argmin_{\Prob(\Rn)} \cE_{\Wab}.
\end{equation}

\begin{corollary}[Narrow convergence of minimizers to unit simplices]\label{C:minimizers to minimizers}
Fix $\beta \ge 2$.  Given $\epsilon>0$, taking $\alpha$ sufficiently large ensures that each
$\displaystyle
 \mab \in \argmin_{\Prob_0(\Rn)} \cE_{\Wab}$ 
satisfies $d_2(\mab, \Mi) <\epsilon$ where $\Mi$ is from \eqref{Mi}. 
\end{corollary}

\noin
{\bf Proof:} 
The set of measures $\mu \in \Prob_0(\Rn)$ satisfying the diameter bound $\diam[\spt \mu] \le R_{ \beta,\beta} := \lim_{\al \searrow \bt} R_{\al,\bt} < \infty$
and with barycenter at the origin is well-known to be $d_2$-compact, e.g.~\cite{Villani03}.   Since $\al >\bt$ implies $\Rab \le R_{\bt,\bt}$, the
corollary becomes a standard consequence of
the $\Gamma$-convergence shown in Lemma \ref{L:Gamma convergence} and the diameter bound \eqref{diameter bound}
as in Theorem~1.21 of \cite{Braides02}.   
\QED
\\

This corollary implies that for $\al$ large enough,  most of the mass of a minimizer $\mab$ lies near the vertices of a unit simplex
(and is approximately equidistributed amongst the $n+1$ vertices).  In view of the Euler-Lagrange condition \eqref{Euler-Lagrange}
the next proposition and its corollary improve this statement to assert that all of the mass of $\mab$ lies near the vertices of a unit simplex.
They rely on the following lemma concerning the potentials of the conjectured optimizers on the 
higher dimensional generalization $\Omega\subset \Rn$ of Reuleaux's triangle:

\begin{lemma}[Unit simplex potentials are minimized only at vertices]\label{L:stratified}
Fix $\beta \ge 2$.
Let $X=\{x_0,x_1,...,x_n\}$ be the set of vertices of a unit $n$-simplex $\Delta^n \subset \R^n$, 
and 
$\Omega := \bigcap_{i=0}^n \overline{B_1({x_i})}$. Define $V: \Omega \subset \R^n \to \R$ by 
\be\label{V2}
 V(x)=-  \sum_{i=0}^n |x-x_i|^\beta.
\end{equation}
Then (a) $X = \argmin_\Omega V$ and (b) when $\beta=2$ then $V$ has no local minima outside $X$.\\
\end{lemma}

\noin{\bf Proof.} (b) Assume $\beta=2$. It is clear that $V$ is strictly concave in ${\rm int} (\Omega)$ so has no local minima there. 
Like the boundary of the simplex $\Delta^n$, which is a stratified space whose strata consist of the relative interiors of unit simplices of 
all lower dimensions,   the boundary of $\Omega$ is a stratified space whose strata consist of open pieces of round spheres of different radii and dimension;
in both cases the zero dimensional strata coincide with the vertices $X$ of $\Delta^n$.
The strategy of our proof is to show strict geodesic concavity of the restriction of $V$ to each of the strata of $\p \Omega$,  which ensures that $V$ cannot
admit local minima except at the zero-dimensional strata.

Given $x^* \in \p \Omega \setminus X$, we will show $V$ cannot attain a local minimum at $x^*$.
By rearranging the indices if necessary, there is $k \in \{1,2,...,n-1\}$ such that
\ba
 |x^*-x_i| =1 & \ \text{for } i  = 0,1,...,k-1, 
 \\ 0 < |x^*-x_i| < 1 &  \ \text{for } i  =k,k+1,...,n. 
 \label{nonbinding vertices}
 \end{align}
Recall that the simplex  $\Delta^{k-1}:=\conv \{x_0,...,x_{k-1}\}$ has radius $r_{k-1} := \sqrt{\frac{k-1}{2k}}$;
take the origin to be its center $\frac1k \sum_{i=0}^{k-1} x_i$ without loss of generality.
We claim the intersection of spheres
\be\label{stratify}
S:=\{x \in \R^n \ | \ |x-x_i|=1\ \ \text{for } i=0,1,...,k-1\}.  
\end{equation}
lies in the subspace of $\R^n$ orthogonal to $\Delta^{k-1}$,
and is in fact the intersection of this subspace $\Sigma := [\Delta^{k-1}]^\perp$ with the sphere of radius $R =  \sqrt{\frac{k+1}{2k}}$ centered at the origin.

Let us establish this claim before completing the proof of the lemma.  
For each $0 \le i< j \le k-1$, the pairwise intersection 
\[
|x - x_i|=1=|x-x_j| 
\]
lies in the hyperplane through the origin orthogonal to $x_i-x_j$; this implies $S \subset \Sigma$.
At each point $x \in S$,  it follows that the vectors $\{x - x_i\}_{i=0}^{k-1}$ are linearly independent.
The implicit function theorem then shows $S$ to be a manifold of dimension $n-k$.  (It cannot be empty since $x^* \in S$.)
For $x \in S$, Pythagoras
yields
$$
|x-0|^2 = |x-x_0|^2 - |0-x_0|^2 = 1 - r_{k-1}^2 = \frac{k+1}{2k} ={R^2  }
$$
whence $S \subset \Sigma \cap \p B_R(0)$.  Since both compact manifolds have the same dimension and the larger of the two is connected,
this inclusion becomes an equality and establishes the claim.

Now $S$ is a round $n-k$ dimensional sphere containing $x^*, x_k,...,x_n$.
Moreover,  \eqref{nonbinding vertices} shows $x^*$ lies in the relative interior of the $n-k$ dimensional manifold-with-boundary $S \cap \p \Omega$.
Choose any constant-speed geodesic curve $\gamma(t)$ valued in $S$  with $\gamma(0)=x^*$, and let $j \in \{k,...,n\}$.  We find
\begin{align*} \frac{d^2}{dt^2}\bigg{|}_{t=0} 
|\gamma(t)-x_j|^2 
&= - 2\frac{d^2}{dt^2}\bigg{|}_{t=0} x_j \cdot \gamma(t) 
\\ &= -2 x_j \cdot \gamma''(0) 
\\ &> 0,
\end{align*}
where the inequality follows  from the facts (i) that $-\gamma''(0)$ is a positive multiple of $x^*$, hence is a linear combination 
 with positive coefficients of $\{ x_k,\ldots, x_n\}$
and (ii) 
$x_i \cdot x_j > 0$ for all $i =k,\ldots, n$ (which follows from the fact that $\langle \hat e_i - c, \hat e_j -c \rangle = \frac1{k}$ 
for the standard simplex in $\R^{n+1}$ using $c=(\frac1k,\ldots,\frac1k,0,\ldots,0)$ in place of the origin).
 When $\beta =2$ this shows the function $t \mapsto V(\gamma(t))$ is strictly concave around $t=0$, 
hence $V$ cannot attain a local minimum at $x^*$, thus proving (b).\\
 
 \noin(a) Now suppose $\beta > 2$.  For each $x \in \Omega$ and $x_i \in X$ we have $|x -x_i|^\beta \le |x - x_i|^2$,
 and the inequality is strict unless $|x-x_i| \in \{0,1\}$.   Thus 
 \[
 V(x) \ge - 
 \sum_{i=0}^n |x-x_i|^2
  \]
 and the inequality is {strict} unless  $x \in X$, where Remark \ref{R:standard simplex} has been used.
Part (b) implies that this lower bound is minimized precisely on~$X$,
hence the same conclusion follows for $V$.
  \QED
 
\begin{proposition}[Convergence of potentials]\label{P:convergence of potentials}
Fix 
$\beta \ge 2$.  Given $r>0$, taking $\alpha$ sufficiently large ensures for each
$\displaystyle
 \mab \in \argmin_{\Prob_0(\Rn)} \cE_{\Wab}
 $ 
there exists $\mi \in \Mi$ such that all minima of $\Vab:=\mab * \Wab$ lie within distance $r$ of $\spt \hat\mu$.
Moreover, $d_2(\mab,\mi) = d_2(\mab,\Mi)$.
\end{proposition}

\noin
{\bf Proof:}
Fix $\mi \in \Mi$ and define $X:=\spt \mi = \{x_0,\ldots,x_n\}$ and its (open) $r$-neighborhood
\[
X_r :=  \bigcup_{i=1}^n B_{r}(x_i).
\]
Given $\delta \in \R$, define
\[
\Omega_\delta := \bigcap_{i=1}^n \overline{B_{1+\delta}(x_i)}.
\]
Note $\Omega_0$ is a strict convexification of $\Delta^n:=\conv(X)$ sharing the same ``vertices'',
and $\Omega_{\pm \delta}$ are slight enlargements and reductions thereof.

By the rotational symmetry of the problem,  it suffices to restrict our attention {to} those minimizers 
$\mab \in \argmin_{\Prob_0(\Rn)} \cE_{\Wab}$ for which $d_2(\mab,\Mi) = d_2(\mab,\mi)$. 
The narrow convergence shown in Corollary~\ref{C:minimizers to minimizers} implies that given $\epsilon>0$,
taking $\al$ large enough ensures that all such minimizers satisfy
\be\label{near equi}
|\mab(B_r(x_i))- \frac1{n+1}|<\frac{\epsilon}{n+1}
\end{equation}
for each $x_i \in X$.
Notice $\Omega_0$ is precisely the set where $\Vib := \hat \mu * \Wib$ is finite,  and the latter is strictly concave on $\Omega_0$,
being a sum of $n+1$ translates of $-\Wb$.
The proof of the proposition requires estimates for the convergence of $\Vab = \mab * \Wab$ to $\Vib$ in three different regions:

\vspace{2mm}
{\em Exterior estimate: Given $\delta>0$ and $R<\infty$, taking $\al$ large enough ensures $\Vab > R$ on $\Rn \setminus \Omega_\delta$.}

\vspace{2mm}

Proof of exterior estimate: For each $y \in \R^n \setminus \Omega_{\delta}$ there is $x \in X$ such that $|x-y| \ge 1+\delta$. 
Note that $\wab(r) := r^\alpha/\alpha - r^\beta/\beta$  converges uniformly  to infinity on $[1+\delta /2, \infty)$ as $\al \to \infty$. 
Now given $\ep < \delta/2$, taking $\al$ sufficiently large ensures that $\mab(B_{\ep}(x)) \approx \frac{1}{n+1}$ within the error $\ep$.  
This implies, taking $\al$ larger if necessary, 
\[ \int_{B_\ep(x)} \Wab (y-z)\, d\mu(z) > 2R  \q \text{for all $y$ with } |y-x| \ge 1+\delta.
\]
On the other hand, since $\Wab \ge -1/2$ we have $\nu * \Wab \ge -1/2$ on $\R^n$ for any nonnegative measure $\nu$ with $\nu(\R^n) \le 1$. Hence we get
\[ \int_{\R^n} \Wab (y-z) \, d\mu(z) > 2R-1/2  \q \text{for all $y$ with } |y-x| \ge 1+\delta.
\]
Since this estimate holds for each $x \in X$, the exterior estimate is established.

\vspace{2mm}

{\em Boundary estimate: Define $A:= \cup_{i=0}^n A_i$, where $A_i$ is the compact neighbourhood of $x_i \in X$ given by
the intersection of spherical annuli
\[
A_i = A_i(\delta,\delta') := \bigcap_{j \ne i} \overline {B_{1+\delta}(x_j)} \setminus B_{1-\delta'}(x_j). 
\]
For small $\delta,\delta'>0$,  we claim taking $\al$ sufficiently large ensures
\be\label{boundary estimate}
\min_{\Omega_\delta \setminus A} \Vab \ge -2 \delta + \min_{\Omega_0 \setminus A} \Vib.
\end{equation}
}
Proof of \eqref{boundary estimate}:
Decompose $\mab = \mu_r + \tilde \mu_r$ into its restriction $\mu_r$ to $X_r$ and its complement.
Taking $\al$ sufficiently large ensures $\tilde \mu_r[\Rn] \le \delta$ according to \eqref{near equi}.
Decompose $\Wab = \Wab^\le + \Wab^>$ into its short range and long range parts as in \eqref{srlr}, noticing $\Wab^> \ge 0$.
 Let $x \in \Omega_{\delta} \setminus A$. Observe that for sufficiently large $\al$,
\begin{align}
\Vab(x)& = (\Wab * \mab) (x)  \nn
\\ &\ge (\Wib^\le * \mab) (x)  \nn
\\ &=(\Wib^\le * \mu_r) (x) + (\Wib* \tilde \mu_r) (x) \nn \\
&\ge (\Wib^\le * \mu_r) (x) -\delta /2    \q \text{since \ $\Wib^\le \ge -1/2$ and $\tilde\mu_r(\R^n) \le \delta$,}  \nn \\
& \ge (\Wib^\le * \hat \mu) (x) -\delta  \q \text{for\ {$\al$ large enough} by Corollary \ref{C:minimizers to minimizers},} \nn \\
& \ge (\Wib^\le * \hat \mu) (y) -2\delta  \q \text{for some} \ y \in \Omega_0 \setminus A, \nn
\end{align}
since  $\Wib^\le$ is 1-Lipschitz, and for each $x \in \Omega_{\delta} \setminus A$, there exists $y \in \Omega_0 \setminus A$ such that $|x-y| \le \delta$ 
{(see figure \ref{bdest})}. 
Taking the infimum over $y\in \Omega_0 \setminus A$ and then over $x \in \Omega_\delta\setminus A$ yields the desired inequality \eqref{boundary estimate}.

\vspace{2mm}
 \begin{figure}
  \includegraphics[width=0.5\linewidth]{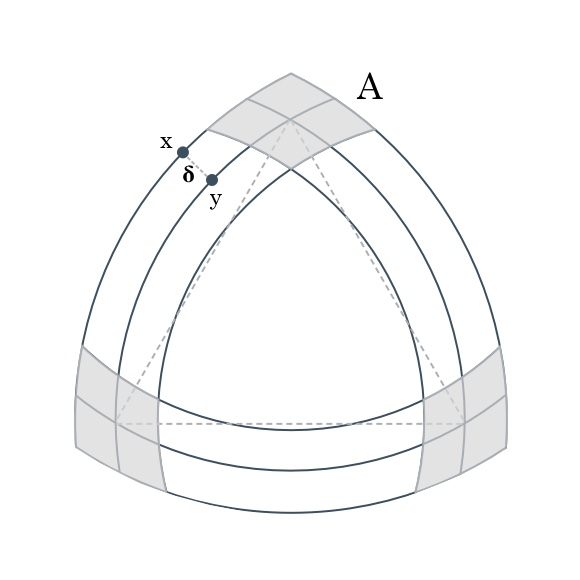}
  \caption{Relevant regions for the estimates.}
  \label{bdest}
\end{figure}
{\em Interior estimate: $\Vab$ converges uniformly to $\Vib$ on $\Omega_{-\delta''}$ for each $\delta''>0$.}

\vspace{2mm}

Proof of interior estimate: Take $\delta>0$ small (e.g. $\delta < \delta'' / 8$), and recall that for sufficiently large $\al$ we have 
$\mab(B_{\delta/2}(x)) \approx \frac{1}{n+1}$ for every $x \in X$ by \eqref{near equi}.
The diameter bound \eqref{diameter bound} then implies, taking $\al$ larger if necessary, that 
\[ \spt(\mu_{\al, \bt}) \subset \Omega_{\delta}.
\]
Note that for every  $x \in \Omega_{-\delta''}$ and $y \in \Omega_{\delta}$, we have $|x-y| \le 1-\delta''/2$. Recall $\wab \to \wib$ uniformly on $[0,1]$ as $\al \to \infty$. 
These facts, plus the narrow convergence of $\mab$ to $\hat \mu$ from Corollary \ref{C:minimizers to minimizers}, imply
\begin{align*}
&\max_{x \in \Omega_{-\delta''}}  | \Vab(x) - \Vib(x) | 
\\ &=\max_{x \in \Omega_{-\delta''}} |(\Wab * \mab)(x) - (\Wib * \hat \mu)(x) |  \nn 
\\ &\le  \max_{x \in \Omega_{-\delta''}} | ((\Wab - \Wib)* \mab)(x)| +  \max_{x \in \Omega_{-\delta''}} |(\Wib * (\mab-\hat \mu))(x) |  
\\ &< \ep \nn
\end{align*}
for $\al$ sufficiently large, given $\ep >0$. This proves the interior estimate.

\vspace{2mm}

{\em Now we prove the proposition.}  Given $r>0$, take $\delta,\delta'>0$ sufficiently small that $A = A(\delta,\delta') \subset X_r$.
Recall that the limiting potential $\Vib$ is continuous and strictly concave on $\Omega_0$,  $+\infty$ outside, 
and attains
its minimum value $\omega=\Vib(x_0)$ precisely on $X$ by Lemma \ref{L:stratified}.
  Notice $f(\delta') = \min_{\Omega_0 \setminus A} \Vib$ is independent of $\delta>0$
and increases continuously with $\delta' \ge 0$ from $f(0)=\omega$. Take $\delta$ smaller if necessary so that 
$2\delta < f(\delta')-\omega$.  For $\al$ sufficiently large the boundary estimate yields
\be\label{boundary consequence}
\min_{\Omega_\delta \setminus A} \Vab > \omega.
\end{equation}
The interior estimate guarantees that by taking $\delta''$ sufficiently small and $\al$ sufficiently large,
we can make $\min_{\Omega_{-\delta''}} \Vab$ as close to $\omega$ as we please ---
 less than \eqref{boundary consequence} in particular.  Taking $\al$ larger if necessary ensures the 
values of $\Vab$ outside $\Omega_\delta$ are all larger than \eqref{boundary consequence}.
In this case the minimum of $\Vab$ can only be attained in $A \subset X_r$.  \QED

\begin{corollary}[Optimizers vanish outside some neighbourhood of a unit simplex]
\label{C:concentration}
Fix $\mi \in \Mi$, $\bt \ge 2$ and $r,\ep \in (0,1/2)$. 
If $\al$ is sufficiently large and $\mu \in \argmin_{\Prob_0(\Rn)} \cE_{\Wab}$ with $d_2(\mu,\mi) = d_2(\mu, \Mi)$
then $\sum_i \mu(B_r(x_i))=1$ and $|\mu(B_r(x_i)) - \frac{1}{n+1}| < \frac{\ep}{n+1}$ for each $x_i \in \spt \mi := \{x_0,\ldots, x_{n+1}\}$.
\end{corollary}

\noin
{\bf Proof:}
The estimate \eqref{near equi} was verified in the course of proving Proposition \ref{P:convergence of potentials},
which also asserts that the potential $V := \mu * \Wab$ is not minimized outside of $X_r := \cup_{x \in \spt \mi} B_r(x)$.
But the Euler-Lagrange equation \eqref{Euler-Lagrange} established  {by} the quartet and trio shows that $\mu$ vanishes outside
$\argmin_\Rn V \subset X_r$. Since $r<1/2$ implies that $X_r$ is a union of $n+1$ disjoint balls,  we conclude 
$\sum_i \mu(B_r(x_i))=1$ as desired.
\QED

\section{Identifying local and global energy minimizers}
\label{S:local}

This section is devoted to the proof of our two main results, Theorems~\ref{T:local} and \ref{T:global},
which identify $d_\infty$-local energy minimizers throughout the mildly repulsive triangle $\al>\bt \ge 2$
and characterize the global energy minimizers for large $\al$ in this range.  The key to both results is the following 
localization theorem based on second variation, 
which allows us to improve on the conclusion of Corollary~\ref{C:concentration}. 
Its proof consists of a comparison 
showing that
if the support of measure $\mu$ lies in a sufficiently small (say $r>0$) neighbourhood of the vertices $X$ of a unit $n$-simplex, 
then for each $x \in X$, the energy of $\mu$ can be reduced by concentrating all of its mass in $B_r(x)$ 
{at the center of mass of the restriction of $\mu$ to this ball.} This is done by establishing a uniformly convex lower bound for the potential $\mu * W$ 
at its minimum in $B_r(x)$, which allows us to estimate the local variance to be zero for any local energy minimizer $\mu$, 
hence all of its mass there to concentrate at a single point. 
A byproduct of this same argument shows the points form a top-dimensional unit simplex.
Thus there are $d_\infty$-local energy minimizers $\mu$ concentrating all of their mass on the vertices of a 
unit simplex (and the mass is nearly equidistributed in the case of a global energy minimizer).  
For the latter case, a comparison with facts we have already proved then allows us to 
remove the adjective `nearly'.

\begin{theorem}[Energetic localization of mass to a unit simplex]\label{T:concentration}
{Fix $m_n \ge \ldots m_1  \ge m_0  >0$ with $\sum_{i=0}^n m_n=1$, $\bt^*>\bt \ge 2$, and the set $X =  \{x_0,x_1,\ldots,x_n\}\subset \R^n$ of vertices of a unit $n$-simplex. If  $0 < \rho \le {m_0m_1}/{m_n^2}$ and
\be\label{astar}
\al^* := \left\{
\begin{array}{ll}
\bt^* + 2(\bt^*-\bt)  & \text{if}\ \bt>2,
\\  \bt^*+  2(\bt^*-\bt) + \rho^{-1} \min\{n,2\}
 & \text{if}\ \bt=2,
\end{array}
\right.
\end{equation}
then there exists
$r=r(\bt^*,\bt,\rho,n)>0$} so that the following holds:
if $\al > \al^*$ and $\mu,\mi \in \Prob(\Rn)$ with $d_\infty(\mu,\mi) \le r$ and 
$\cE_{\Wab}(\mu) \le \cE_{\Wab}(\mi)$, and if $\mi$ vanishes outside $X$ but
$m_i=\mi[\{x_i\}]>0$ for each $i =0,1,\ldots,n$, then 
$\mu$ is a rotated 
translate of $\mi$.
\end{theorem}

\noin
{\bf Proof:}
First assume $\bt^*>\beta>2$ and $\al>\al^*=\bt^*+2(\bt^*-\bt)$ and $0< \rho \le m_0m_1/m_n^2$ and set $2\eta := \al^*-\bt^*$.  
For $r >0$ small enough (to be determined later, and independently of $\al$),  
let $\mu,\mi \in \Prob(\Rn)$ satisfy all the hypotheses of the theorem,  so that
\begin{equation}\label{micentration}
\mi=\sum_{i=0}^n m_i\delta_{x_i}. 
\end{equation}
Let  $\mu_i$ be the restriction of $\mu$ to $B_r({x_i})$.
For $r < 1/2$, the hypothesis $d_\infty(\mu,\mi)<r$ implies $\mu_i(\R^n)=m_i$
and $\mu=\sum_{i=0}^n \mu_i$.

Let us {abbreviate} $W=\Wab$ and $w=\wab$,
and consider the energy difference $F(\mu) :=\cE_W(\mu)-\cE_W(\mi) \le 0$
(which is non-positive by hypothesis).  With {$i,j = 0,1,\ldots, n$} we 
observe
\begin{align}\label{Fdef}
F(\mu) ={\sum_{i=0}^n  \bigg[ \int (\mu_i * W) d\mu_i + \sum_{j \neq i}\iint\big( W(x-y)-w(1) \big) d\mu_j(x) d\mu_i(y) \bigg]. }
\end{align}
Let $\nu_i:=\mu_i / m_i$ be the normalization of $\mu_i$. Since $\beta>2$, given any $\epsilon>0$
there exists $r=r(\epsilon)>0$ such that $W(x) \ge -\epsilon|x|^2$ in $B_r(0)$. Hence for every $i$,
\begin{align}\label{self energy}
\int (\mu_i *W) d\mu_i
&\ge - {m_i^2} \epsilon \iint
|x-y|^2 d\nu_i(x) d\nu_i(y) 
\\&= -2m_i^2 \epsilon {\rm \Var}(\nu_i)  \nn 
\end{align}
where ${\rm \Var}(\nu_i)$ is the variance  \eqref{variance} of $\nu_i$. 

Since $\al>\al^*$, the computation
\begin{align*}
[\wab(s)-\wab(1)] - [\wsb(s)-\wsb(1)] &= \was(s)-\was(1) 
\\&\ge 0
\end{align*}
shows $\wab(x)-\wab(1)$ to be a non-decreasing function of $\al$.
Noting $\wsb''(1) = \al^*-\beta > 2\eta>0$,  taking $s_0>0$ small enough 
(depending on $\al^*$ and $\beta$ but not $\alpha$) 
yields
\begin{align*}
\wab(s) - \wab(1) 
&\ge \wsb(s)-\wsb(1)
\\&\ge 
\eta(s-1)^2 {\q \text{on} \q [1-s_0,1+s_0].}
\end{align*} 
Now define $\zeta(z):=  (|z|-1)^2$. 
Since $|x_i-x_j|=1$, for $r$ small enough that $z_i \in B_r(x_i)$ and $z_j\in B_r(x_j)$ implies $| |z_i-z_j| - 1| \le s_0$, we have 
\begin{align}\label{exchange energy}
&\sum_{i=0}^n \sum_{j \ne i} \iint \big( W(z_i-z_j)-w(1) \big) d\mu_j(z_j) d \mu_i(z_i)
\\&  \ge 
\eta m_0 m_1 \sum_{i=0}^n \sum_{j \ne i} \iint \zeta(z_i-z_j)  d\nu_j(z_j) d\nu_i(z_i). \nn
\end{align}

To estimate the integrand, let $y_i := \bar x(\nu_i)$ be the barycenter \eqref{barycenter} of $\nu_i$. Let
$v_i := z_i  - y_i$, $\Delta v_{ij} := v_i-v_j$,  ${\Delta \hat y_{ij}} := \frac{y_i- y_j}{|y_i -y_j|}$, etc.  Then
\begin{align*}
|\Delta z_{ij}| 
&= \sqrt{|\Delta y_{ij}|^2 + 2 \langle \Delta y_{ij}, \Delta v_{ij} \rangle  + |\Delta v_{ij}|^2}  
\\ & = |\Delta y_{ij}|  +  \langle \Delta \hat y_{ij}, \Delta v_{ij} \rangle + O(|\Delta v_{ij}|^2)
\end{align*}
whence {$|v_i| \le 2r$, $| \Delta v_{ij} | \le 4r$, and $| |\Delta y_{ij}| -1 | \le 2r$ imply
 \begin{align*}
    \zeta(\Delta z_{ij}) =&  
(|\Delta y_{ij}| - 1)^2 + 2(|\Delta y_{ij}|-1) \langle \Delta \hat y_{ij}, \Delta v_{ij}\rangle + \langle \Delta \hat y_{ij}, \Delta v_{ij} \rangle^2
\\ & {+ O(r|\Delta v_{ij}|^2  )} 
  \end{align*}
}and
\begin{align}\label{binomial error}
&\iint \zeta(\Delta z_{ij}) d\nu_i(z_i) d\nu_j(z_j) 
 = 
(|\Delta y_{ij}| - 1)^2 
\\ & +\iint [\langle \Delta \hat y_{ij}, \Delta v_{ij} \rangle^2 + {O(r|\Delta v_{ij}|^2)}  ]  d\nu_i(z_i) d\nu_j(z_j);
\nn
\end{align}
here the error term does not depend on any parameters except through its argument.

From
$
\Delta v_{ij} = v_i -v_j
$
we compute
\begin{align}\label{quadratic1}
\iint \langle \Delta \hat y_{ij}, \Delta v_{ij} \rangle^2 d\nu_i(z_i) d\nu_j(z_j)
= 
 &\int {\langle \Delta \hat y_{ij}, v_i \rangle^2} d\nu_i(z_i)
+ \int \langle \Delta \hat y_{ij}, v_j  \rangle^2 d\nu_j(z_j)
\nn\end{align}
and 
\[
\sum_{j=1}^n   \int \langle \Delta \hat y_{0j}, v_0  \rangle^2 d\nu_0(z_0)
=  \int \langle v_0 , A_0  v_0 \rangle  d\nu_0(z_0)
\]
where the matrix $A_0$ is given by

\be\label{simplex operator}
A_0 = \sum_{j=1}^n \frac{y_0 -y_j}{|y_0 -y_j|} \otimes \frac{y_0 - y_j}{|y_0 - y_j|}. 
\end{equation}
In case $y_j = x_j$ for all $j=0,1,\ldots,n$,  a direct calculation using a scaled copy of the standard $n$-simplex 
\eqref{standard simplex} in $\R^{n+1}$
shows 
$A_0$ has $1/2$ as an eigenvalue of multiplicity $n-1$ and $\frac{n+1}2$ as a simple eigenvalue.
In this case $A_0 \ge \Id/\min\{n,2\}$ in the sense that the difference of the two matrices is non-negative definite,
where $\Id$ is the $n \times n$ identity matrix.  More generally, $|y_j-x_j| \le r$ for all $j$, from which it follows that
$A_0 \ge \frac{1 + O(r)}{\min\{n,2\}}\Id$.  Thus
\begin{equation}\label{0variance bound}
\sum_{j=1}^n   \int \langle \Delta \hat y_{0j}, v_0 \rangle^2 d\nu_0(z_0)
\ge \frac{\Var(\nu_0)} {\min\{2,n\}} (1+O(r))
\end{equation}
and
\begin{align*}
\sum_{i=0}^n \sum_{j \ne i} \iint \langle \Delta \hat y_{ij}, \Delta v_{ij} \rangle^2 d\nu_i(z_i) d\nu_j(z_j)
&\ge  \frac{2+O(r)}{\min\{n,2\}}  \sum_{i=0}^n 
\Var(\nu_i). 
\end{align*}
Noting also
\begin{align*}
\iint |\Delta v_{ij}|^2 d\nu_i d\nu_j = \Var(\nu_i)  + \Var(\nu_j),
\end{align*}
from \eqref{binomial error} we deduce
\begin{align*}
&\sum_{i=0}^n \sum_{j\ne i}  \iint \zeta({\Delta z_{ij}})  d\nu_i(z_i) d\nu_j(z_j)
\\&\ge  \sum_{i=0}^n   \left[
\frac{2+O(r)}{\min\{n,2\}} \Var(\nu_i) + \sum_{j\ne i} (|\Delta y_{ij}|-1)^2
\right]. 
\end{align*}
Recalling \eqref{exchange energy}, choose
$\la < \frac2 {\min\{n,2\}}$ and take $r>0$ smaller if necessary (depending on $\la$)
to obtain
\begin{align*}
& {\sum_{i=0}^n \sum_{j \ne i} \iint  (W(\Delta z_{ij})  - w(1))d\mu_i(z_i) d\mu_j(z_j)}
\\&\ge  \eta m_0m_1 \sum_{i=0}^n \left[  \la \Var(\nu_i) + \sum_{j \ne i}  (| \Delta y_{ij}|-1)^2 \right], 
\end{align*}
where the new constant absorbs the $O(r)$ term. With \eqref{Fdef}--\eqref{self energy} this gives
\be\label{variance coef}
F(\mu)
\ge  \eta m_0m_1 \sum_{i=0}^n \left[ (\la -  \frac{2 \epsilon }{\eta \rho} ) {\rm \Var}(\nu_i)  + \sum_{j \ne i} ( |\Delta y_{ij}|-1)^2 \right].
\end{equation}
For $0< 2\epsilon < \eta \la \rho$, choosing $r$ small enough (depending on $(\al^*,\bt,\rho)$ and our choice of $(\la,\epsilon)$), 
validates the above arguments, forcing the coefficient of $\Var(\nu_i)$ in the summation above 
to be positive for all $i$.
Now since $F(\mu) \le 0$ by assumption, this leads to the conclusion
${\rm \Var}(\nu_i)=0$ and $|y_i - y_j|=1$ for each distinct {$i,j \in \{0,1,\ldots, n\}$.} Thus 
\begin{equation} 
\mu=\sum_{i=0}^n m_i\delta_{y_i} 
\end{equation}
and the barycenters $y_i$ form a unit $n$-simplex.
We infer $\mu$ is obtained from $\mi$ by a slight translation and/or rotation, in view of Remark \ref{R:standard simplex}.  
This concludes the case $\bt>2$.

Now suppose $\bt=2$.  In this case, no matter how small $r>|x|>0$ is,  we will not have $W(x) \ge - \epsilon |x|^2$ unless
$\epsilon \ge 1/2$.  However, 
we may  take $\epsilon = 1/2$ in the preceding argument.
To compensate,  we will need ${\eta \rho} 
> \frac1\la >\frac{\min\{n,2\}}{2}$,
or equivalently
\begin{align}\label{yyy}
2\eta = \al^*-\bt^* > \frac{\min\{n,2\}}{ \rho}.
\end{align}
This follows from our choice \eqref{astar} of $\al^*$.
Now the foregoing argument implies the same conclusion.
\QED

As a first application, we show that all measures on the vertices of a unit $n$-simplex are strict $d_\infty$-local energy
minimizers in the following sense. 

\begin{definition}[Strict $d_\infty$-local energy minimizer]
Given $\al>\bt \ge 2$,  a measure $\hat \mu \in \Prob(\Rn)$ is a strict $d_\infty$-local energy minimizer of $\cE_{\Wab}$ if there exists $r>0$
such that $d_\infty(\mu,\hat \mu) <r$ implies $\cE_{\Wab}(\mu) \ge \cE_{\Wab}(\mi)$,  and equality holds only if $\mu$ is a 
rotated translate of $\mi$.
\end{definition}
Theorem \ref{T:local} follows directly from: 
\begin{corollary}[All distributions over unit simplex vertices are $d_\infty$-local energy minimizers] \label{C:strict local minimizer}
Fix $m_n \ge \cdots \ge m_1 \ge m_0 >0$ summing to one, 
$\bt^*>\bt > 2$, $\al^*$, and $r$  
as in Theorem \ref{T:concentration}
and the set $X :=\{x_0,x_1,\ldots,x_n\} \subset \R^n$ of vertices of a unit $n$-simplex.
If  $\mi \in \Prob(\Rn)$ satisfies \eqref{micentration},
then  $\al>\al^*$ implies $\mi$ is a strict $d_\infty$-local minimizer of $\cE_{\Wab}$ on $\Prob(\Rn)$;
 i.e.~there exists $r>0$ such that rotated translates of $\hat \mu$ uniquely minimize $\cE_{\Wab}$
among measures $\mu\in \Prob(\Rn)$ satisfying $d_\infty(\mu, \hat \mu) < r$.
\end{corollary}

\noin
{\bf Proof:}  First assume $X=\spt \mi$. Under the hypotheses of the corollary, if  $d_\infty(\mu,\mi)<r$ but
$\cE_{\Wab}(\mu) < \cE_{\Wab}(\mi)$,  Theorem \ref{T:concentration} asserts 
that $\mu$ is a rotated translate of $\mi$,
contradicting the invariance of $\cE_{\Wab}$ under such symmetries.  This contradiction forces 
the desired conclusion: $\cE_{\Wab}(\mu) \ge \cE_{\Wab}(\mi)$.  If $\cE_{\Wab}(\mu) = \cE_{\Wab}(\mi)$,  
Theorem \ref{T:concentration} asserts $\mu$ is a rotated 
translate of $\mi$.  
 \QED
 
Theorem \ref{T:concentration} shows that we can choose $\al$ as close to $\bt$ as we please in this corollary unless $\bt=2$, and
even when $\bt=2$ we need not choose $\al$ very large unless $m:=\min_i m_i$ is very small.
When $\bt=2$ we can reformulate the theorem as an estimate for a phase transition threshold:

\begin{corollary}[Centrifugal threshold for strict $d_\infty$-local minimizers]\label{C:locastar}
Take $\bt = 2$, $X\subset \Rn$ and $\hat \mu \in \Prob(X)$ as in Theorem \ref{T:concentration},
with the mass $m_i$ of $\hat \mu$ at each vertex $x_i$ in $X$ satisfying bounds 
$0 < m_0 \le m_1 \le \cdots \le m_n$. Then there exists a smallest 
$\al_{\Delta^n}^{loc}=\al_{\Delta^n}^{loc}(m_0,\ldots,m_n) \ge { 2}$ such that
for each $\al>\al_{\Delta^n}^{loc}$, the measure $\hat \mu$ is a strict $d_\infty$-local energy minimizer.  Moreover,
\be\label{aloc}
\al_{\Delta^n}^{loc} 
\le 2+ \frac{m_n^2 \min\{n,2\}}{m_0m_1}.
\end{equation}
\end{corollary}

\noin{\bf Proof:}
As in Corollaries \ref{C:main} and \ref{C:strict local minimizer}, the existence of $\al_{\Delta^n}^{loc} \in [\bt,\infty]$ is obvious;  Theorem \ref{T:concentration} 
implies the bound \eqref{aloc}.
\QED

\begin{remark}[{Sharpness}]\label{R:1d}
When $(\beta,n)=(2,1)$,  taking $m_0=m \le \frac12$ and $m_n=m_1 = 1- m$, results of 
 Kang, Kim, Lim and Seo \cite{KangKimLimSeo19p} show the bound \eqref{aloc} becomes an equality
$\al_{\Delta^1}^{loc} := 1 + \frac1m$.
{When $\beta = 2 \le n$ but $m_0=m_n$,  the linear instability found for $\al<4$ by
Sun, Uminsky and Bertozzi
\cite{SunUminskyBertozziJMP12} {strongly suggests} that equality also holds in the bound 
$\al_{\Delta^n}^{loc}(\frac1{n+1},\ldots,\frac1{n+1})\le4$.
}
\end{remark}

A last but not least application will be to derive our main result on global minimizers, Theorem \ref{T:global}, restated here for the reader's convenience:

\begin{corollary}[Optimizers equidistribute over the vertices of a unit simplex] 
Given $\bt \ge 2$, taking $\al$ sufficiently large and $\Miw$ from \eqref{Mi} ensures 
\[ \argmin_{\Prob(\R^n)} \cE_{\Wab} = \Miw.\]
\end{corollary}

\noin
{\bf Proof:}
Fix $0<\epsilon =\beta^*-\beta<1$, let $\al^*$ and {$ r=r(\bt^*,\bt,(\frac{1-\ep}{1+\epsilon})^2,n)>0$ be as in Theorem \ref{T:concentration}}.
Taking $\alpha> \al^*$ large enough and $\mu \in \argmin_{\Prob(\Rn)} \cE_{\Wab}$,  
Corollary \ref{C:concentration} yields $\mu$ vanishing outside the neighbourhood of radius $r$ around the vertex set $X=\{x_0,\ldots,x_n\}$ 
of a unit $n$-simplex,  and with $m_i := \mu[B_r(x_i)]$ satisfying $|(n+1)m_i -1 | \le \epsilon$ for each {$i = 0,1,\ldots,n$.}
The measure $\mi\in \Prob(X)$ from \eqref{micentration} then satisfies $d_\infty(\mu,\mi) \le r$ and the choice of $\mu$ ensures 
$\cE_{\Wab}(\mu) \le \cE_{\Wab}(\mi)$.
Theorem~\ref{T:concentration} now asserts $\mu$ is a translated rotation of $\mi$,
hence $\mi$ is also a global energy minimizer.
However,  for measures $\nu \in \Prob(\Rn)$ vanishing outside $X$,
we have $\cE_{\Wab}(\nu) = (1-\frac\bt\al)\cE_{\Wib}(\nu)$.
Corollary \ref{C:hard confinement} shows the latter functional is 
minimized precisely by translations of the measures in $\Mi = \Miw \cap \Prob_0(\Rn)$.
Thus we conclude some translate of $\mi$ (and of $\mu$) lies in $\Mi$ as desired,  or equivalently that $\mu[\{x\}] = \frac1{n+1}$ for all $x \in \spt \mu$.
 \QED
 
 \appendix
 
 \section{Isodiametry, variance, and regular simplices}
  \label{S:Jung}

 Our variational characterization of the unit simplex, Theorem \ref{T:isodiametric variance bound},
 was discovered using convex analysis and duality in \cite{LimMcCann19p}.  However,
 it turns out to be closely related
 to a classical result of Jung \cite{Jung01},
for which a modern proof can be found in Danzer, Gr\"unbaum and Klee \cite{DanzerGruenbaumKlee63}:

\begin{theorem}[Jung] \label{T:Jung}
Let $K \subset \R^n$ be compact with $\diam(K)=1$. Then $K$ is contained in a closed ball of radius $r_n=\sqrt{\frac{n}{2n+2}}$. Moreover, $K$ contains the vertices of a unit $n$-simplex 
unless it lies in some smaller ball.
\end{theorem}

In our companion work we showed that our characterization implies Jung's 
theorem \cite{LimMcCann19p}.   In this appendix we show instead that our 
characterization follows from Jung's theorem,  so that the two results are 
in some sense equivalent.
We are grateful to an anonymous seminar participant for 
drawing our attention to Jung's work, 
and to Tomasz Tkocz  \cite{Tkocz19p} who subsequently observed independently from us that our characterization could be inferred using Jung's theorem.  
Let us begin with an elementary geometric result based on
Lemma \ref{L:stratified}, which
concerns higher dimensional generalizations $\Omega \subset \R^n$ of Reuleaux's triangle
and tetrahedron.

\begin{lemma}[On Reuleaux simplices]\label{L:Reuleaux}
If $\Delta \subset \R^n$ is the set of vertices of a unit $n$-simplex centered at $z \in \R^n$ 
and $\Omega := \cap_{x \in \Delta} \overline{B_1(x)}$, then 
$\Delta = \Omega \cap \p B_{r_n}(z)$ where $B_r(x)$ denotes the ball of radius $r$ centered at $x$.
\end{lemma}

\begin{proof}
Let $\Delta = \{x_0,\ldots,x_n\} \subset \R^n$ be {the vertices of} a unit $n$-simplex centered at 
$z = \frac1{n+1}\sum x_i$.
Any vectors $y_0,\ldots,y_n$ in a Hilbert space $H$ satisfy
$$
|\sum_{i=0}^n y_i|^2 + \sum_{0 \le i<j \le n}  |y_i - y_j|^2 = (n+1) \sum_{i=0}^n |y_i|^2.
$$
Given an arbitrary point $x \in \Omega := \cap_{x \in \Delta} \overline{B_1(x)}$,
 taking $y_i=\frac1{n+1}(x-x_i)$ and $H=\R^n$ the identity above yields
 $$
 |x - z|^2 + \frac{n}{2n+2} = \frac1{n+1}  \sum_{i=0}^n |x-x_i|^2.
 $$
 Estimating the right hand side with Lemma \ref{L:stratified}(a) yields
 $ |x-z|^2 \le r_n^2, $
 with equality if and only if $x \in \Delta$. Thus $\Omega \subset \overline{B_{r_n}(z)}$
 and $\Delta = \Omega \cap \p B_{r_n}(z)$ as desired.
\end{proof}

\noin{\bf Proof of Theorem \ref{T:isodiametric variance bound} using Theorem \ref{T:Jung}.} 
The representation \eqref{standard simplex} shows the {vertices of a} standard $n$-simplex of diameter $\sqrt{2}$ lies
on a unique sphere of radius $r_n \sqrt{2}$; thus the {vertices of a} unit $n$-simplex lies on a (unique) sphere of radius $r_n$.
Assume $d=1$ without loss of generality hereafter.  
Any probability measure $\mu^*$ which assigns mass 
$1/(n+1)$ to each vertex of a unit $n$-simplex therefore has the desired variance $r_n^2$.  
Conversely,  let $\mu \in \Prob(\R^n)$
have support $K=\spt \mu$ with $\diam[K] \le 1$.  Jung's theorem then asserts $K$ is enclosed
by a sphere $S=\p B_{r}(z)$ of radius $r \le r_n$ centered at some $z \in \R^n$, and that
$r<r_n$ unless $K$ contains a unit $n$-simplex.  The familiar computation
\begin{equation}\label{variance around a non-center}
 \Var(\mu) + |\bar x(\mu) - z|^2 = \int_K |x-z|^2 d\mu(x)  \le r^2 \le r_n^2
\end{equation}
shows $\Var(\mu) \le r_n^2$.  We conclude equidistribution
$\mu^*$ over the vertices of the unit $n$-simplex 
has maximal variance subject to the unit diameter constraint on its support.
Also, \eqref{variance around a non-center} shows $\Var(\mu)<r_n^2$ unless $\bar x(\mu)=z$ and $r=r_n$.
Thus $\mu$ has smaller variance than $\mu^*$ unless $K$
contains {the vertices of} a unit $n$-simplex $\Delta:= \{x_0,\ldots, x_n\} \subset K$.

We henceforth assume $\Var(\mu)=r_n^2$, so
$\Delta \subset K = \spt \mu$ and $\bar x(\mu)=z$.  
From $\Var(\mu)=r_n^2$ 
and $\spt \mu \subset \overline{B_{r_n}(z)}$ we conclude the full mass of $\mu$ lies at distance $r_n$ from its barycenter $z=\bar x(\mu)$, i.e. $K\subset S= \partial B_{r_n}(z)$.
On the other hand,  $\diam(K) \le 1$ and $\Delta \subset K$ implies 
$K \subset \Omega$ where $\Omega := \cap_{i=0}^n \overline{B_1(x_i)}$.  
Lemma \ref{L:Reuleaux}  therefore implies $K=\spt \mu \subset S \cap \Omega= \Delta$.  Now there is
a familiar bijection between the convex hull $\conv(\Delta)$ and convex combinations of its vertices,
c.f.~Remark 2.5 \cite{LimMcCann19p}.
The only convex combination of the vertices of $\Delta$ having barycenter at $z$ assigns equal
weights $1/(n+1)$ to each vertex. From $\bar x(\mu)=z=\frac1{n+1} \sum_{i=0}^n x_i$ we deduce 
$\mu = \frac1{n+1} \sum_{i=0}^n \delta_{x_i}$ as desired.  
\QED

\end{document}